\documentclass{llncs}
\usepackage[utf8]{inputenc}
\usepackage{amsfonts}
\usepackage{amsmath}
\usepackage[english]{babel}
\usepackage[T1]{fontenc}
\usepackage[dvips]{graphicx}
\usepackage{pstricks,pst-plot}
\usepackage{lastpage}

\newcommand{\gf}{\mathbb{F}_2}

\title{Cancellation-free circuits: An approach for proving
superlinear lower bounds for linear Boolean operators\thanks{Partially supported by the 
Danish Council for Independent Research, Natural Sciences.}
}
\titlerunning{Cancellation-free circuits}

\author{Joan Boyar
\and
Magnus Gausdal Find
\institute{Department of Mathematics and Computer Science\\ University of
Southern Denmark\\}
\email{\\\{joan, magnusgf\}@imada.sdu.dk}
}
\authorrunning{Joan Boyar\and Magnus G. Find}
\begin{document}

\maketitle

\begin{abstract}
We continue to study the notion of \emph{cancellation-free linear
circuits}.
We show that every matrix can be computed by a cancellation-free circuit,
and almost all of these are at most a constant factor larger than
the optimum linear circuit that computes the matrix.
It appears to be easier to prove statements about
the structure of cancellation-free linear circuits
than for linear circuits in general. We prove two
nontrivial superlinear lower bounds.
We show that a cancellation-free linear circuit computing
the $n\times n$ Sierpinski gasket
matrix must use at least $\frac{1}{2}n\log n$ gates, and that this is
tight.
This supports a conjecture by Aaronson.
Furthermore we show that a proof strategy
for proving lower bounds on monotone circuits can be almost 
directly converted to prove lower bounds on cancellation-free linear
circuits.
We use this together with a result from
extremal graph theory due to Andreev 
to prove a lower bound of $\Omega(n^{2-\epsilon})$
for infinitely many $n\times n$ matrices
for every $\epsilon>0$ for.
These lower bounds for concrete matrices are almost optimal
since all matrices can be computed with
$O\left(\frac{n^2}{\log n}\right)$
gates.
\end{abstract}

\section{Introduction and Known Results}
Let $\gf$ be the Galois field of order $2$, and let $\gf^n$
be the $n$-dimensional vector space over $\gf$.
A Boolean function $f: \gf^n \rightarrow \gf^m$ is said
to be linear if there exists a Boolean $m\times n$ matrix $A$
such that $f(\mathbf{x})=A\mathbf{x}$ for every $\mathbf{x}\in \gf^n$.
This is equivalent of saying that $f$ can be computed using
only XOR gates.

An  \emph{XOR-AND circuit} $C$ is a directed acyclic graph.
There are $n+1$ 
nodes with in-degree $0$, called the \emph{inputs} one
of these is the constant value $1$.
All other nodes have in-degree $2$ and are called \emph{gates}.
Every gate is labeled either $\oplus$ (XOR) or $\wedge$ (AND).
There are $m$ gates which are called the \emph{outputs}; these are
labeled $y_1,\ldots,y_m$.
The value of a gate labeled $\wedge$ is the product of its inputs
(children), and
the value of a gate labeled $\oplus$ is the sum
of its two children (addition in $\gf$, denoted $\oplus$).
The circuit $C$, with
inputs $\mathbf{x}=(x_1,\ldots ,x_n)$,
\emph{computes} the $m\times n$ matrix $A$
if the output vector computed by $C$,
$\mathbf{y}=(y_1,\ldots , y_m)$,
satisfies $\mathbf{y}=A\mathbf{x}$.
In other words, output $y_i$ is defined by the $i$th row of the matrix.
The \emph{size} of
a circuit $C$, denoted $|C|$, is the number of gates in $C$. 
For simplicity, we will let $m=n$ unless otherwise is explicitly stated.
A circuit is \emph{linear} if every gate is labeled $\oplus$.

For relatively dense matrices,
computing all the rows independently gives $\Theta(n)$ gates
for each output, that is a circuit of size $\Theta(n^2)$.
It follows from a theorem by Lupanov
\cite{juknabook,Lupanov1965})
that this upper bound can be improved.
\begin{theorem}[Lupanov]
\label{lupanovthm}
  Every $n\times n$ matrix can be computed using a circuit of size
\[
(1+o(1))\frac{n^2}{\log n}.
\]
\end{theorem}
A counting argument shows, that this is asymptotically tight.
In fact, the vast majority of
matrices require this number of gates up to a constant factor.
Despite this fact, there is no known concrete family of matrices requiring
superlinear size \cite{juknabook}.

Another, but related circuit model is the one where we allow unbounded
fan-in and arbitrary gates
(that is gates computing \emph{any} predicate are allowed), but
require bounded depth.
The circuit complexity of such a circuit is the number of wires. 
Here the lower bound situation is a little better;
Alon, Karchmer and Wigderson \cite{AlonKW90} showed in 1990 
 that a particular
family of matrices requires $\Omega(n\log n)$ wires for linear circuits
in this model.
This has recently been improved by Gál et al. \cite{GalHKPV11} 
 who have proven that a concrete infinite family of matrices require
$\Omega\left(n\left(\frac{\log n}{\log \log n}\right)^2\right)$
wires when computed in depth 2.
Recently Drucker \cite{Drucker11a} gave a survey of the strategies
used for proving lower bounds on wire complexity
for general (not necessarily linear) Boolean operators in bounded
depth, and the limitations of these.

Returning to the circuit model with bounded fan-in,
the situation is even worse for
general Boolean predicates. Here we know
by a seminal result by Shannon
\cite{Shannon1949synthesis,Wegener87},
that almost every function
requires $\Omega(2^n/n)$
gates,
but again no superlinear bound is known for a concrete
family of functions.
A popular, and essentially the only known,
technique for proving non-trivial linear lower bounds is the technique of 
\emph{gate-elimination}. The key idea when using gate
elimination is to
set some of the inputs to constant values, arguing that a certain
number of gates 
get ``eliminated'' and that this results in a
function inductively assumed  to have
a certain size.
Gate elimination was first used by Schnorr \cite{Schnorr74} to
prove a $2n$ lower bound, and later improved by Paul
\cite{Paul77} and again by Blum \cite{Blum84} who in
1984 presented a $3n$ lower bound 
for a family of functions when using the full
binary basis.
This is still the best concrete lower bound known \cite{juknabook}.
For a description of the gate-elimination method see the
survey of Boppana and Sipser
\cite{BoppanaS90} or the essay by Blum \cite{DBLP:conf/birthday/Blum09}.
In both of these it is mentioned that it is unlikely that the
gate elimination method will ever yield superlinear lower bounds.

In the case of general Boolean functions there are a
number of functions
conjectured to have superlinear size, examples include any
$NP$-complete language. 
For linear operators there are, as far as the authors know,
only few families of matrices conjectured to have superlinear size.
One of these include the Sierpinski gasket
matrix, (Aaronson, personal communication
and \cite{cstheorystackexchange})
described later in this paper.

One proof strategy for proving lower bounds is to prove
lower bounds for a restricted circuit model, and to prove that
sizes of circuits computing a function in the restricted circuit model
are not too much larger than in the original model.
This was essentially the motivation for looking at monotone circuits.
In \cite{Razborov1985lower},
Razborov gave a superpolynomial lower bound for the Clique function for
monotone circuits.
The hope was at that time, that the monotone circuit complexity
was polynomially related to general Boolean circuit complexity.
This was disproven by Razborov in \cite{razborov1985permanent},
showing that the gap was superpolynomial. For more details, see e.g.
\cite{BoppanaS90}.

\section{Cancellation-free Linear Circuits}
\label{cancelfreelinear}
For linear circuits, the value computed by every gate is the parity function
of some subset of the $n$ variables.
That is, the output of every gate $u$ can be considered
as a vector $\kappa(u)$ in the vector space $\gf^n$,
where $\kappa(u)_i=1$ if and only
if $x_i$ is a term in the parity function computed by the gate $u$.
We call $\kappa(u)$ the \emph{value vector} of $u$, and
for input variables define
$\kappa(x_i)=e^{(i)}$, that is the unit vector having the
$i$th coordinate $1$ and all other $0$.
It is clear by definition that if a gate $u$ has the two children $w,t$, then
$\kappa(u)=\kappa(w)\oplus~\kappa(t)$, where $\oplus$ denotes coordinate wise
addition in $\gf$.
We say that a linear circuit is \emph{cancellation-free} if for
every pair of gates $u,w$ where  $u$ is
an ancestor of $w$ then $\kappa(u)\geq \kappa(w)$,
where $\geq$ denotes the usual coordinatewise partial order.
That is, if $x_i$ is a term in a gate $w$ it is a term in all subsequent
gates.
The intuition behind this
is that if this condition is satisfied, the circuit never
exploits the fact that in $\gf$, $a\oplus a=0$.
That is, things do not ``cancel out'' in the circuit.
By definition, it is clear that any linear operator can be computed
by a cancellation-free circuit. The proposition comes directly
from the definition of cancellation-free

\begin{proposition}
\label{equiprop}
  The following are equivalent:
  \begin{itemize}
  \item $C$ is cancellation-free
  \item For every pair of vertices $v_1,v_2$ in $C$, there do
not exist two disjoints paths in $C$ from $v_1$ to $v_2$
\item For every $v$ where $\kappa (v)_i=0$ there is no path
from $x_i$ to $v$
\item $C$ does not contain the triangle $K_3$
as an undirected minor
  \end{itemize}
\end{proposition}

The notion cancellation-free was introduced by Boyar
and Peralta in \cite{boyar2010combinational,BoyarMP08}.
The paper concerns straight line program for computing
linear forms, which is equivalent to the model studied in
this paper. 
They proved that the problem of finding shortest linear circuits
for linear operators is NP hard, even when restricted
to cancellation-free circuits.
They also noticed that most heuristics for constructing small linear circuits
never exploit the cancellation property. Then, they
constructed a gate minimizing heuristic that uses cancellation.

\section[Relationship]{Relationship
Between Cancellation-free Linear Circuits and General Linear Circuits}
Boyar and Peralta proved in \cite{boyar2010combinational} that
there exists an infinite family of matrices
where the sizes of cancellation-free circuits computing them
are at least $\frac{3}{2}-o(1)$ times larger than the optimum. 
We call this ratio the \emph{cancellation ratio}, $\rho(n)$.
We can strengthen the lower bound
to $2$ using a surprisingly simple matrix.
This construction is originally due to Svensson \cite{johnny}.

\begin{theorem}
There exists an infinite family of matrices such that any
cancella\-tion-free circuit computing them must have size 
  $2-o(1)$ times larger than the optimum. Thus $\rho(n)\geq 2-o(1)$
\end{theorem}\begin{proof}
Consider the $n\times n$ matrix:
\[
\begin{pmatrix}
  0     & 1 & 1 & 1 & 1 & \ldots & 1 & 1\\
  1     & 1 & 0 & 0 & 0 & \ldots & 0 & 0\\
  1     & 1 & 1 & 0 & 0 & \ldots & 0 & 0\\
  1     & 1 & 1 & 1 & 0 & \ldots & 0 & 0\\
        &   &   &  &\vdots&      &   &  \\
  1     & 1 & 1 & 1 & 1 & \ldots & 1 & 1\\
\end{pmatrix}
\]

If one allows cancellation this matrix can be computed
by a circuit of size $n$, by first computing $x_1\oplus x_2$ to 
obtain $y_2$. For $3\leq j\leq n$, adding $x_j$
to $y_{j-1}$ gives $y_j$. Thus, we use $n-1$ gates to compute
$y_2,\ldots y_n$. After that we can obtain $y_1$ with
one gate since $y_1=y_n\oplus x_1$.

Consider any cancellation-free linear circuit $C$ computing the matrix.
Let the set $S$ contain the gate computing $y_1$ and all its (noninput) predecessors.
Clearly $|S|\geq n-2$ since it is the
sum of $n-1$ terms.

Notice that because $C$ is cancellation-free, none of the gates
in $S$ can compute any of the output values $y_2,\ldots y_n$.
Therefore for every $j>1$ we need at least one gate to compute $y_j$.
Thus one needs $n-1$ extra gates for this part. This adds up to
$2n-3$. And the ratio is therefore $\frac{2n-3}{n}$
proving the theorem.\qed
\end{proof}

It turns out that for almost every matrix, the cancellation ratio is
constant.

\begin{lemma}
If cancellation is allowed almost every $n\times n$
matrix needs $\frac{n^2}{4\log n}-o(\frac{n^2}{\log n})$ gates to be computed.
\end{lemma}
\begin{proof}
The number of $n\times n$ matrices is $2^{n^2}$.
Since there are two inputs to each of the $M$ gates, 
and each of the $n$ outputs are either the output from a gate or
an input (or zero), the
number of circuits with $n$ inputs, $n$ outputs and $M$ gates is at
most 
\[
(n+M)^{2M}(n+M+1)^n/M!
\]
Taking the logarithm one gets
\[
  {2M}\log (n+M)+ n\log(n+M+1) -\log(M!)
\]
Recalling that $\log(M!)=M\log M - O(M)$, for sufficiently
large $n$, $n< M$:
\[
 {2M}\log (2M)+ M\log(2M) - M\log M + O(M)=2M\log M+O(M)
\]
so the number of distinct circuits is at most $2^{2M\log M+O(M)}$.
For $0<\epsilon<1$ the number of matrices that can be computed
with $M=(1-\epsilon)\frac{1}{4}n^2/\log n$ gates is at most
\[
2^{2M\log M+O(M)}\leq 2^{(1-\epsilon)n^2+o(n^2)}.
\]
That is, the fraction of matrices
not computable, is at least
\[
1-\frac{2^{(1-\epsilon)n^2+o(n^2)}}{2^{n^2}}.
\]
Since this limit tends to $1$
almost every matrix has circuit
size at least $\frac{n^2}{4\log n}-o\left(\frac{n^2}{4\log n}\right)$.
\qed
\end{proof}

We will now show that the construction in the proof of 
Theorem~\ref{lupanovthm}
produces a circuit that is cancellation-free.
Before stating the lemma and its proof we will need a definition of
rectangular decompositions:
Given a Boolean $n\times n$ matrix $A$, the Boolean 
matrices $B_1,\ldots,B_k$
constitute a rectangular decomposition if
$A=B_1+B_2+\ldots + B_k$ where addition is \emph{over the reals}
and every $B_i$ has rank $1$. We say that the \emph{weight} of $B_i$ is
the number of nonzero columns plus the number of nonzero rows.
The weight of a rectangular decomposition is the sum of the
weights of the $B_i$'s.
Lupanov showed in \cite{Lupanov1965} (see also \cite{juknabook})
that every $n\times n$ matrix admits a rectangular decomposition
of weight $(1+o(1))\frac{n^2}{\log n}$.
\begin{lemma}
  Every $n\times n$ matrix can be computed by a cancellation-free
linear circuit of size $(1+o(1))\frac{n^2}{\log n}$.
\end{lemma}
\begin{proof}
Let the Boolean $n\times n$ matrix $A$ be arbitrary. Consider the rectangular
decomposition $B_1,\ldots,B_k$ assumed to exist by Lupanov's theorem.
For each $i$ let $c_i$ ($r_i$) denote the number of nonzero
columns (rows) in $B_i$.
Add for each $B_i$ the inputs corresponding to the nonzero columns, using
$c_i-1$ gates.
Call the result $s_i$.
Now each output
is a sum of $s_i$'s. For each $y_j$, add these $s_i$'s. In total, this 
takes at most $\sum_i r_i$ gates.
The total number of gates is at most
\[
\sum_i (c_i-1) + \sum_j r_j \leq  (1+o(1))\frac{n^2}{\log n} 
\]
Since the addition $B_1+\ldots +B_k$ in the
the rectangular decomposition is over the reals, the circuits is
cancellation-free.
\qed
\end{proof}
Combining the two lemmas we get the following:
\begin{theorem}
  For almost every matrix, the cancellation ratio, $\rho (n)$, is constant.
\end{theorem}

\section[Sierpinski]{
Lower Bound on the Size of Cancellation-free Circuits
Computing the Sierpinski Gasket Matrix.
}
In this section we will prove that the $n\times n$ Sierpinski gasket
matrix needs $\frac{1}{2}n\log n$ gates when computed by a linear
cancellation-free circuit, and that this suffices.  

Suppose some subset of the input variables are restricted to the value $0$.
Now look at the resulting circuit.
Some of the gates will now compute the value $z=0\oplus w$.
In this case, we say that the gate is eliminated
since it no longer does any computation. The situation can be even more
extreme, some gate might ``compute'' $z=0\oplus 0$.
In both cases, we can remove the gate from the circuit, and
forward the input if necessary (if $z$ is an output
gate, $w$ now outputs the result). In the second
case, the parent of $z$ will get eliminated, so the
effect might cascade.
For any subset of the variables,
there is a unique set of gates that become eliminated when setting
these variables to $0$.

The Sierpinski gasket matrix is defined recursively as:
\[
S_0=
\begin{pmatrix}
  1
\end{pmatrix}
\]

\[
S_{k+1}=
\begin{pmatrix}
  S_{k} & 0\\
  S_{k} & S_{k}\\
\end{pmatrix}
\]
In all of the following let $n=2^k$, and let $S_k$ be the
$n\times n$ Sierpinski gasket matrix.
First we need a fact about $S_k$:

\begin{proposition}
For every $k$ the determinant of the Sierpinski gasket matrix is
$1$.
In particular the $2^k$ rows in $S_k$ 
are linearly independent.
\end{proposition}
\begin{proof}
The determinant of an augmented matrix is given by
the formula:
\[
det(S_{k+1})=det\begin{pmatrix}S_{k} & 0\\ S_{k} & S_k\end{pmatrix}
= det(S_k)det(S_k)=1
\]
\qed
\end{proof}

\begin{theorem}
\label{sierpinskilower}
For every $k\geq 2$, any cancellation-free circuit that computes 
the $n\times n$ Sierpinski gasket matrix has size at least
$\frac{1}{2}n\log_2 n$.
\end{theorem}

\begin{proof}
The proof is by induction on $k$. For the base case, look at
the $2\times 2$ matrix $S_1$.
This clearly needs at least
$\frac{1}{2}2\log_2 2=1$ gate.

Suppose the statement is true for some $k$, now
look at the $2n\times 2n$  matrix $S_{k+1}$. 
Denote the output gates $y_1,\ldots,y_{2n}$
and the inputs $x_1,\ldots,x_{2n}$.
Partition the gates of $C$ into three disjoint sets,
$C_1,C_2$ and $C_3$ defined as follows:
\begin{itemize}
\item $C_1$: The gates having only inputs from $x_1,\ldots,x_n$ and
$C_1$.
Equivalently the gates not reachable from inputs $x_{n+1},\ldots,x_{2n}$.
\item $C_2$: The gates in $C-C_1$ that are not eliminated
when inputs $x_1,\ldots,x_n$ are set to $0$.
\item $C_3$: $C-(C_1\cup C_2)$. That is, the gates in $C-C_1$ that 
do become eliminated when inputs $x_1,\ldots,x_n$ is set to $0$.
\end{itemize}
Obviously $|C|=|C_1|+|C_2|+|C_3|$.
We will now give lower bounds on the sizes of $C_1$, $C_2$, and $C_3$.

\paragraph{$C_1$:}
Since the circuit is cancellation-free, the
outputs $y_1,\ldots,y_n$ and all their predecessors are in $C_1$.
By the induction hypothesis, $|C_1|\geq \frac{1}{2}n\log_2 n$.

\paragraph{$C_2$:}
Since the gates in $C_2$ 
are note eliminated when $x_1=x_2=\ldots =x_n$, they
compute $S_k$ on the
inputs $x_{n+1},\ldots,x_{2n}$. By the induction hypothesis
 $|C_2|\geq \frac{1}{2}n\log_2 n$.

\paragraph{$C_3$:} The goal is to prove that this set has size
at least $n$. 
Let $\delta(C_1)$ be the set of arcs
from $C_1\cup \{x_1,\ldots,x_n\}$ to $B=C_2\cup C_3$. We first prove that
\begin{equation}
  \label{eq:cutatleastc3}
|C_3|\geq |\delta(C_1)|  
\end{equation}
By definition, all gates in $C_1$ attain the value $0$ when
$x_1,\ldots,x_n$ are set to $0$.
Let $(v,w)\in \delta(C_1)$ be arbitrary.
Since $v\in C_1\cup \{x_1,\ldots , x_n\}$,
$w$ becomes eliminated, so $w\in C_3$. Every $u\in C_3$
can only have one child in $C_1$, since
no gate in $C_3$ can have two children in $C_1$.
So $|C_3|\geq |\delta(C_1)|$.

We now show that $|\delta(C_1)|\geq n$.
Let the endpoints of $\delta(C_1)$ in $C_1$ be $e_1,\ldots , e_p$
and let their corresponding value vectors be $v_1,\ldots , v_p$.


Now look at the value vectors of
the output gates $y_{n+1},\ldots,y_{2n}$.
For each of these, the first vector consisting of the
first $n$ coordinates must be in $span(v_1,\ldots ,v_p)$, but
the dimension of $S_k$ must is $n$, so $p\geq n$.


We have that $|C_3|\geq |\delta(C_1)|\geq n$, so
\[
|C|=|C_1|+|C_2|+|C_3| \geq
\frac{1}{2}n\log_2 n + \frac{1}{2}n\log_2 n
+ n = \frac{1}{2}(2n)\log_2(2n).
\]
\qed
\end{proof}

It turns out that this is tight.
\begin{proposition}
  The Sierpinski matrix can be computed
by a cancellation-free circuit
using $\frac{1}{2}n\log_2 n$ gates.
\label{constructiveSierpinskiUpperBound}
\end{proposition}
\begin{proof}
This is clearly true for $S_2$. Assume that $S_k$ can be computed using
$\frac{1}{2}n\log_2 n$ gates. Consider the matrix $S_{k+1}$. Construct the circuit
in a divide and conquer manner; construct recursively on variables $x_1,\ldots,x_n$
and $x_{n+1},\ldots, x_{2n}$. This gives outputs $y_1,\ldots,y_n$.
After this use $n$ operations to finish the outputs $y_{n+1},\ldots y_{2n}$.
This adds up to exactly $\frac{1}{2}(2n)\log_2 2n$. \qed
\end{proof}

\section{Stronger Lower Bounds}

In \cite{mehlhorn1979some}, Mehlhorn proved lower
bounds on monotone circuits
for computing
``Boolean sums''.
The same proof strategy can be used to prove lower bounds on
cancellation-free linear circuits.
For a matrix $A$, denote by $c_f(A)$ the smallest cancellation-free
linear circuit that
computes $A$, and $|A|$ as the number of $1$'s in $A$.
Let $K_{a,b}$ be the complete bipartite graph with $a$ vertices
in one vertex set and $b$ in the other.

\begin{theorem}
\label{mehlhornish}
  Let $M$ be an $n\times n$ matrix. Interpret $M$ as a
  vertex adjacency matrix for a bipartite graph in the
  natural way. If this graph does not contain 
$K_{h+1,k+1}$ for constants
$h,k$ 
  then $|c_f(M)|\in \Omega(|M|)$.  
\end{theorem}
\begin{proof}
  Consider the class of cancellation-free linear circuits where all sums
  of at most $k$ variables are available for free.
  Let $\tilde{c_f}(M)$ be smallest of such circuits computing $M$.
  Obviously $|c_f(M)|\geq |\tilde{c_f}(M)|$.
Since all sums of at most $k$
variables are
  available for free, anything computed at a gate in
  $\tilde{c_f}(M)$ is a sum of at least $k+1$
  variables.
  Since the circuit is cancellation-free,
  for a gate $u$ in $\tilde{c_f}(M)$, its
  value vector will never decrease, hence the value vector of a
  successor to $u$ will
  have $1$ on the $k+1$ coordinates that $u$'s value vector has.
  In particular, since the matrix does not contain $K_{h+1,k+1}$, this
  means that any gate $u$ in $\tilde{c_f}(M)$ can have a path
  to at most $h$ outputs.

  For a fixed row $i$, the cost of computing it
  is at least
  \[
    |M_i|/k-1.
  \]
  And since a gate has a path to at most $h$ outputs,
  if we sum over all rows we count
  each gate at most $h$ times. So the total size of
  $\tilde{c_f}(M)$ is at least
  \[
  \sum_{i} (|M_i|/k-1)/h \in \Omega(|M|)
  \]
\qed
\end{proof}

Now, proving lower bounds for linear cancellation-free circuits
is reduced to the problem of finding dense bipartite graphs not containing
$K_{h+1,k+1}$. This problem is known as the Zarankiewicz problem.

\begin{corollary}
For any $\epsilon>0$, there exists a concrete family of
matrices that requires $\Omega(n^{2-\epsilon})$ gates when computed
by a cancellation-free linear circuit.
\label{stronglowerbound}
\end{corollary}
\begin{proof}
In \cite{andreev86}, Andreev gave for every $\epsilon>0$ an explicit
construction for an infinite family of bipartite graphs
with $2n$ nodes and $n^{2-\epsilon}$ edges that does not contain the
subgraph $K_{h+1,k+1}$ where $h$ and $k$ only depend on $\epsilon$.
Using this construction together with Theorem \ref{mehlhornish} gives
the desired result.\qed
\end{proof}

It should be noted that Brown \cite{brown1966graphs}
gave a simpler construction of
a family of graphs with $\Theta(n)$ vertices and $\Theta(n^{5/3})$ edges
not containing $K_{3,3}$.
Also,  Koll\'{a}r et al. \cite{KollarRS96} 
gave a construction
similar to Andreev's, but where the functions $h,k$ grow slower than
in Andreev's construction.

\section{Conclusion and Open Problems}
What is the value of $\rho(n)$?
If for some $\delta>0$, $\rho(n)\in O(n^{1-\delta})$,
Corollary \ref{stronglowerbound} provides an unconditional 
superlinear lower bound for a concrete family of matrices.

In the proof of Theorem \ref{sierpinskilower},
we did not use the cancellation-free property as extensively
as we did in the proof of Theorem \ref{mehlhornish}.
We only used that there
is no path from $x_{n+1},\ldots,x_{2n}$ to the outputs $y_1,\ldots,y_n$.
Another strategy to prove an unconditional lower bound on the
size of circuits computing the Sierpinski matrix could be to
prove that for any optimal circuit no such path exists.
Then the theorem would follow, even with cancellations.

\bibliographystyle{splncs03}
\bibliography{cancellationlowerbound}

\end{document}